\documentclass[11pt,letterpaper]{article}

\def\showauthornotes{0}
\def\showtableofcontents{0}
\def\showkeys{0}
\def\showdraftbox{0}
\def\showcolorlinks{1}
\def\usemicrotype{1}
\def\showfixme{0}

\usepackage{etex}

\usepackage[l2tabu, orthodox]{nag}

\usepackage{xspace,enumerate}

\usepackage[dvipsnames]{xcolor}

\usepackage[american]{babel}

\usepackage{mathtools}
\usepackage{amsmath,amsfonts}

\usepackage{amsthm}

\newtheorem{theorem}{Theorem}[section]
\newtheorem*{theorem*}{Theorem}

\newtheorem*{proposition*}{Proposition}
\newtheorem{lemma}[theorem]{Lemma}
\newtheorem*{lemma*}{Lemma}
\newtheorem{corollary}[theorem]{Corollary}
\newtheorem*{conjecture*}{Conjecture}

\newtheorem*{fact*}{Fact}

\newtheorem*{hypothesis*}{Hypothesis}

\theoremstyle{definition}
\newtheorem{definition}[theorem]{Definition}

\theoremstyle{remark}

\newtheorem*{claim*}{Claim}
\newtheorem{remark}[theorem]{Remark}
\newtheorem*{remark*}{Remark}

\newtheorem*{observation*}{Observation}

\usepackage[
letterpaper,
top=1in,
bottom=1in,
left=1in,
right=1in]{geometry}

\ifnum\showkeys=1
\usepackage[color]{showkeys}
\fi

\ifnum\showcolorlinks=1
\usepackage[
pagebackref,
colorlinks=true,
urlcolor=blue,
linkcolor=blue,
citecolor=OliveGreen,
]{hyperref}
\fi

\ifnum\showcolorlinks=0
\usepackage[
pagebackref,
colorlinks=false,
pdfborder={0 0 0}
]{hyperref}
\fi

\usepackage{prettyref}

\newcommand{\savehyperref}[2]{\texorpdfstring{\hyperref[#1]{#2}}{#2}}

\newrefformat{eq}{\savehyperref{#1}{\textup{(\ref*{#1})}}}
\newrefformat{eqn}{\savehyperref{#1}{\textup{(\ref*{#1})}}}
\newrefformat{lem}{\savehyperref{#1}{Lemma~\ref*{#1}}}
\newrefformat{def}{\savehyperref{#1}{Definition~\ref*{#1}}}
\newrefformat{thm}{\savehyperref{#1}{Theorem~\ref*{#1}}}
\newrefformat{cor}{\savehyperref{#1}{Corollary~\ref*{#1}}}
\newrefformat{cha}{\savehyperref{#1}{Chapter~\ref*{#1}}}
\newrefformat{sec}{\savehyperref{#1}{Section~\ref*{#1}}}
\newrefformat{app}{\savehyperref{#1}{Appendix~\ref*{#1}}}
\newrefformat{tab}{\savehyperref{#1}{Table~\ref*{#1}}}
\newrefformat{fig}{\savehyperref{#1}{Figure~\ref*{#1}}}
\newrefformat{hyp}{\savehyperref{#1}{Hypothesis~\ref*{#1}}}
\newrefformat{alg}{\savehyperref{#1}{Algorithm~\ref*{#1}}}
\newrefformat{rem}{\savehyperref{#1}{Remark~\ref*{#1}}}
\newrefformat{item}{\savehyperref{#1}{Item~\ref*{#1}}}
\newrefformat{step}{\savehyperref{#1}{step~\ref*{#1}}}
\newrefformat{conj}{\savehyperref{#1}{Conjecture~\ref*{#1}}}
\newrefformat{fact}{\savehyperref{#1}{Fact~\ref*{#1}}}
\newrefformat{prop}{\savehyperref{#1}{Proposition~\ref*{#1}}}
\newrefformat{prob}{\savehyperref{#1}{Problem~\ref*{#1}}}
\newrefformat{claim}{\savehyperref{#1}{Claim~\ref*{#1}}}
\newrefformat{relax}{\savehyperref{#1}{Relaxation~\ref*{#1}}}
\newrefformat{red}{\savehyperref{#1}{Reduction~\ref*{#1}}}
\newrefformat{part}{\savehyperref{#1}{Part~\ref*{#1}}}

\newcommand{\Sref}[1]{\hyperref[#1]{\S\ref*{#1}}}

\usepackage{nicefrac}

\ifnum\usemicrotype=1
\usepackage{microtype}
\fi

\ifnum\showauthornotes=1
\newcommand{\Authornote}[2]{{\sffamily\small\color{red}{[#1: #2]}}}
\newcommand{\Authornotecolored}[3]{{\sffamily\small\color{#1}{[#2: #3]}}}
\newcommand{\Authorcomment}[2]{{\sffamily\small\color{gray}{[#1: #2]}}}
\newcommand{\Authorstartcomment}[1]{\sffamily\small\color{gray}[#1: }

\newcommand{\Authorfnote}[2]{\footnote{\color{red}{#1: #2}}}
\newcommand{\Authorfixme}[1]{\Authornote{#1}{\textbf{??}}}
\newcommand{\Authormarginmark}[1]{\marginpar{\textcolor{red}{\fbox{\Large #1:!}}}}
\else
\newcommand{\Authornote}[2]{}
\newcommand{\Authornotecolored}[3]{}
\newcommand{\Authorcomment}[2]{}
\newcommand{\Authorstartcomment}[1]{}

\newcommand{\Authorfnote}[2]{}
\newcommand{\Authorfixme}[1]{}
\newcommand{\Authormarginmark}[1]{}
\fi

\ifnum\showfixme=0

\fi

\usepackage{boxedminipage}

\newcommand{\Esymb}{\mathbb{E}}

\DeclareMathOperator*{\E}{\Esymb}

\newcommand{\textparen}[1]{\text{(#1)}}

\ifx\because\undefined
\newcommand{\because}[1]{\textparen{because #1}}
\else
\renewcommand{\because}[1]{\textparen{because #1}}
\fi

\newcommand{\mcom}{\,,}

\newcommand\bdot\bullet

\ifx\mathds\undefined 
\DeclareMathOperator{\Ind}{\mathbb{I}}
\else
\DeclareMathOperator{\Ind}{\mathds 1}}
\fi

\DeclareMathOperator{\poly}{poly}

\newcommand{\N}{\mathbb N}
\newcommand{\R}{\mathbb R}

\newcommand{\cP}{\mathcal P}
\newcommand{\cQ}{\mathcal Q}

\newcommand{\cS}{\mathcal S}

\ifnum\showdraftbox=1
\newcommand{\draftbox}{\begin{center}
  \fbox{%
    \begin{minipage}{2in}%
      \begin{center}%
          \Large\textsc{Working Draft}\\%
        Please do not distribute%
      \end{center}%
    \end{minipage}%
  }%
\end{center}
\vspace{0.2cm}}
\else
\newcommand{\draftbox}{}
\fi

\let\epsilon=\varepsilon

\numberwithin{equation}{section}

\newcommand\MYcurrentlabel{xxx}

\newcommand{\MYstore}[2]{%
  \global\expandafter \def \csname MYMEMORY #1 \endcsname{#2}%
}

\newcommand{\MYload}[1]{%
  \csname MYMEMORY #1 \endcsname%
}

\newcommand{\MYnewlabel}[1]{%
  \renewcommand\MYcurrentlabel{#1}%
  \MYoldlabel{#1}%
}

\newcommand{\MYdummylabel}[1]{}

\newcommand{\torestate}[1]{%
  \let\MYoldlabel\label%
  \let\label\MYnewlabel%
  #1%
  \MYstore{\MYcurrentlabel}{#1}%
  \let\label\MYoldlabel%
}

\newcommand{\restatetheorem}[1]{%
  \let\MYoldlabel\label
  \let\label\MYdummylabel
  \begin{theorem*}[Restatement of \prettyref{#1}]
    \MYload{#1}
  \end{theorem*}
  \let\label\MYoldlabel
}

\newcommand{\restatelemma}[1]{%
  \let\MYoldlabel\label
  \let\label\MYdummylabel
  \begin{lemma*}[Restatement of \prettyref{#1}]
    \MYload{#1}
  \end{lemma*}
  \let\label\MYoldlabel
}

\newcommand{\restateprop}[1]{%
  \let\MYoldlabel\label
  \let\label\MYdummylabel
  \begin{proposition*}[Restatement of \prettyref{#1}]
    \MYload{#1}
  \end{proposition*}
  \let\label\MYoldlabel
}

\newcommand{\restatefact}[1]{%
  \let\MYoldlabel\label
  \let\label\MYdummylabel
  \begin{fact*}[Restatement of \prettyref{#1}]
    \MYload{#1}
  \end{fact*}
  \let\label\MYoldlabel
}

\newcommand{\restate}[1]{%
  \let\MYoldlabel\label
  \let\label\MYdummylabel
  \MYload{#1}
  \let\label\MYoldlabel
}

\newcommand{\addreferencesection}{
  \phantomsection
  \addcontentsline{toc}{section}{References}
}

\let\origparagraph\paragraph
\renewcommand{\paragraph}[1]{\origparagraph{#1.}}

\allowdisplaybreaks

\usepackage{paralist}

\usepackage{comment}

\let\citet\cite

\theoremstyle{definition}

\DeclareUrlCommand\email{}

\title{On the Bit Complexity of Sum-of-Squares Proofs}

\author{%
\normalsize
Prasad Raghavendra\thanks{UC Berkeley, \protect\email{raghavendra@berkeley.edu}. Supported by NSF Career Award, NSF CCF-1407779 and the Alfred. P. Sloan Fellowship.}
\and
Ben Weitz\thanks{UC Berkeley, \protect\email{bsweitz@cs.berkeley.edu}. Supported by an NSF Graduate Research Fellowship (NSF award no DGE 1106400).}
}

\begin{document}

\maketitle

\draftbox

\thispagestyle{empty}

\begin{abstract}
	It has often been claimed in recent papers that one can find a degree $d$ Sum-of-Squares proof if one exists via the Ellipsoid algorithm. In \cite{odonnell17}, Ryan O'Donnell notes this widely quoted claim is not necessarily true. He presents an example of a polynomial system with bounded coefficients that admits low-degree proofs of non-negativity, but these proofs necessarily involve numbers with an exponential number of bits, causing the Ellipsoid algorithm to take exponential time. In this paper we obtain both positive and negative results on the bit complexity of SoS proofs. 

	First, we propose a sufficient condition on a polynomial system that implies a bound on the coefficients in an SoS proof.  We demonstrate that this sufficient condition is applicable for common use-cases of the SoS algorithm, such as \textsc{Max-CSP}, \textsc{Balanced Separator}, \textsc{Max-Clique}, \textsc{Max-Bisection}, and \textsc{Unit-Vector} constraints. 
	
	On the negative side, O'Donnell asked whether every polynomial system containing Boolean constraints admits proofs of polynomial bit complexity. We answer this question in the negative, giving a counterexample system and non-negative polynomial which has degree two SoS proofs, but no SoS proof with small coefficients until degree $\Omega(\sqrt{n})$.
\end{abstract}

\clearpage

\ifnum\showtableofcontents=1
{
\tableofcontents
\thispagestyle{empty}
 }
\fi

\clearpage

\setcounter{page}{1}

\section{Introduction}
\label{sec:intro}

The Sum of squares (SoS) proof system is a versatile and powerful approach to certifying polynomial inequalities. 
SoS certificates can be shown to underly a vast number of algorithms in combinatorial optimization. 
On the one hand, SoS certificates hold the promise of yielding algorithms that possibly refute the notorious unique games conjecture \cite{BBHKSZ12, barak2011rounding,GuruswamiS11}.  
On the other hand, a flurry of recent works have applied SoS proofs to develop algorithms for problems ranging from constraint satisfaction problems to tensor problems.  

To illustrate sum of squares certificates, let us consider the example of the
\textsc{Balanced Separator} problem.  Here we are given a graph $G = (V,E)$ and the goal is to find a balanced cut $(S,\overline{S})$ with the minimum number of crossing edges.
Like many problems in combinatorial optimization, it can be reformulated as a low-degree polynomial optimization problem.  
Specifically if we associate $\{0,1\}$ variables $\{x_1,\ldots,x_n\}$ for the vertices of the graph $G$ then we can rewrite the \textsc{Balanced Separator} problem as follows,
\begin{align*}
\text{ Minimize } \ \ \  \sum_{(i,j) \in E} (x_i-x_j)^2  \ \ \ \ \text{ subject to } \left\{ x_i^2 = x_i \forall i\ , \ \frac{n}{3} \leq \sum_i x_i \leq \frac{2n}{3} \right\}
\end{align*} 
Here the constraint $x_i^2 = x_i$ ensures $x_i \in \{0,1\}$ while the inequalities enforce the condition that the cut is balanced. More generally, a low-degree polynomial optimization is of the form
\begin{align*}
&\text{ Minimize } r(x) \text{ subject to } \\ &\text{ equalities } \cP = \left\{ p_i(x) = 0 | i \in [n] \right\} \text{ and inequalities } \cQ = \left\{ q_i(x) \geq 0 | i \in [m] \right\}
\end{align*} 

An SoS certificate of a lower bound $ r(x) \geq \theta$ is given by a polynomial identity of the form
\[ r(x) - \theta  = \sum_{i} h_i(x)^2 + \sum_{i \in [m]} \left(\sum_{j}^{t_i} s_j^2(x) \right) \cdot q_i(x) + \sum_{i \in [n]} \lambda_i(x) p_i(x) \]
Notice that for all $x$ satisfying the equalities $\cP$ and the inequalities $\cQ$, the right hand side of the above identity is manifestly non-negative, thereby certifying that $r(x) \geq \theta$.  The degree of the SoS certificate is the maximum degree of the polynomials involved, i.e., $d = \max\{\deg h_i^2, \deg s_j^2 q_i, \deg \lambda_i p_i\}$.

The main appeal of SoS certificates for polynomial optimization is that the existence of a degree $d$ SoS certificate can be formulated as the feasibility of a semidefinite program (SDP).
This is the degree $d$ SoS relaxation first introduced by Shor \cite{shor1987class}, and expanded upon by later works of Nesterov \cite{nesterov2000squared}, Grigoriev and Vorobjov \cite{grigoriev2001complexity}, Lasserre \cite{lasserre2000optimisation,lasserre2001global}  and Parrilo \cite{parrilo2000structured}.
(see, e.g., \cite{laurent2009sums,barak2014sum} for many more details).

The degree $d$ SoS SDP has $n^{O(d)}$ variables, and if the coefficients of $p$ and $q$ are reasonably bounded (smaller than $2^{n^{O(d)}}$), the resulting SDP has a compact description of size $n^{O(d)}$.
From this, several works including those by the authors, asserted that the resulting feasibility SDP can be solved in time $n^{O(d)}$ using the Ellipsoid algorithm.

In a recent work, O'Donnell \cite{odonnell17} observed that this often repeated claim is far from true.
Specifically, O'Donnell exhibited systems of polynomial inequalities with bounded coefficients such that only degree $2$ SoS certificates of non-negativity involve coefficients that are doubly exponential in size.
Thus all SoS certificates need an exponential number of bits to represent and consequently, the ellipsoid algorithm will incur an exponential running time.

As pointed out by O'Donnell, the issue at hand here is not just that of additive error in the solution, i.e., the difference between testing feasibility and near-feasibility.  
Indeed, semidefinite programming via the ellipsoid algorithm can only test feasibility up to a very small additive error.
However, in a majority of applications of SoS SDP relaxations in combinatorial optimization, the variables in the underlying polynomial system are explictly bounded (also known as Archimedian).
Specifically, these include constraints such as $\{ x_i^2 \leq 1 | i \leq [n]\}$, which yield explicit bounds on the values of the variables.
In these settings, if there is an approximate SoS certificate for $r(x) \geq \theta$, then there exists a proper SoS certificate for a slightly weaker lower bound $r(x) \geq \theta - o(1)$.
Therefore, additive error incurred in semidefinite programming can often be traded off for a slightly weaker objective value.
The issue highlighted by O'Donnell is far more serious in that the coefficients of the SoS certificate are too large -- thereby directly affecting the runtime of the ellipsoid algorithm.

On a positive note, O'Donnell shows that a polynomial system whose only constraints are the Boolean constraints $\{x_i^2 = x_i | i \in [n]\}$ always admit SoS certificates with polynomial bit complexity.  
He proceeds to ask whether all polynomial systems that include boolean constraints, potentially among others, always admit bounded SoS certificates.

\subsection{Our Results}

In this work, we further explore the issue of bit complexity of SoS proofs, and obtain both positive and negative results.

First, we present an easily verifiable and broadly applicable set of sufficient conditions under which a polynomial optimization problem has small SoS certificates.
Roughly speaking, we show that polynomial systems with {\it rich} sets of solutions have bounded SoS certificates of non-negativity.
Consider a system consisting of polynomial equalities $\cP$ and inequalities $\cQ$.  Our approach consists of looking for assignments $S$ satisfying three criteria (see \prettyref{def:nice} and \prettyref{thm:main} for formal statements).  
\begin{theorem}
Assume $(\cP, \cQ, S)$ satisfies:
\begin{enumerate}
\item The assignments $S$ robustly satisfy the inequalities in $\cQ$.  
\item The polynomial calculus proof system is both complete and efficient over $S$.  In other words, all degree $d$ polynomial identities over $S$ can be derived using a degree $O(d)$ polynomial derivation from the equalities $\cP$.
\item The assignments $S$ are {\it spectrally rich} in that smallest non-zero eigenvalue of their covariance matrix is at least $2^{-\poly(n^d)}$. 
\end{enumerate}
Then if $r$ has a degree $d$ proof of non-negativity from $\cP$ and $\cQ$, it also has a degree $O(d)$ proof of non-negativity with coefficients bounded by $2^{\poly(n^d)}$.
\end{theorem}

We demonstrate the broad applicability of the above set of sufficient conditions by using them to show upper bounds on bit complexity for \textsc{Max-CSP}, \textsc{Max-Clique}, \textsc{Matching}, \textsc{Balanced Separator}, \textsc{Max-Bisection}, and optimization over the unit sphere.  In each case, the above sufficient conditions can be verified easily. 

The above set of sufficient conditions are widely applicable in combinatorial optimization, wherein the polynomial system is typically a relaxation of a well-known set of integer solutions.  
In such a setup with integer solutions, we observe in \prettyref{sec:nicespaces} that spectral richness is an immediate consequence of the discrete nature of the set of solutions.
Therefore, in all these setups, the only non-trivial thing to verify is the efficiency of the polynomial calculus proof system.

The work of O'Donnell \cite{odonnell17} exhibited a polynomial system with bounded coefficients which admitted degree $2$ SoS certificate, whose coefficients were necessarily doubly-exponential.
However, the variables in this polynomial system were not all boolean, i.e. did not have the $x_i^2 = x_i$ constraint.
In fact, O'Donnell asked whether every polynomial system with boolean constraints admits a small SoS proof.
Moreover, the polynomial system in \cite{odonnell17} admits a degree $4$ SoS certificate with small bit complexity.  
This opens up the possibility that one can effectively reduce the bit-complexity by raising the degree of the proof.
For instance, if a system admits a degree $d$ SoS certificate then does it always admit a degree $2^d$ SoS certificate with small bit complexity (even under boolean constraints)?
Unfortunately, we refute both of the above possibilities by exhibiting a counterexample.
Formally, we show the following:

\begin{theorem}\label{thm:counter}
	There exists a system of quadratic equations on $n$ variables such that
	\begin{itemize}
	\item The system includes the equation $x_i^2 - x_i = 0$ for each $i \in [n]$.
	\item There exists a polynomial with a degree $2$ SoS certificate of non-negativity, albeit with doubly exponentially large coefficients.
	\item No SoS certificate of degree $d \leq \sqrt{n}$ has coefficients smaller than $\Omega\left(\frac{1}{n^d}\cdot 2^{\exp(\sqrt{n})}\right)$.
	\end{itemize}	
\end{theorem}


\section{Preliminaries}
\label{sec:prelims}
For a set of real polynomials $\cP = \{p_1, p_2, \dots, p_m\}$, we denote their generated ideal in $\R[x]$ by $\langle \cP\rangle$ or $\langle p_1, \dots, p_m\rangle$. We will be working with systems of polynomial constraints, and we will use the $\cP$ to denote the equality constraints, and $\cQ$ to denote the inequality constraints, i.e. $p(x) = 0$ and $q(x) \geq 0$ for $p \in \cP$ and $q \in \cQ$. We will usually use $S$ for the set of points satisfying these constraints. We use $\R[x]_d$ for the set of polynomials of degree at most $d$, and we write ${\bf v}_d$ for the vector of polynomials whose entries are the elements of the usual monomial basis of $\R[x]_d$. Similarly, we use ${\bf v}(\alpha)_d$ for the vector of reals whose entries are the entries of ${\bf v}$ evaluated at $\alpha$. We usually omit the dependencies on $d$ as it is clear from context.

\subsection{Polynomial Proofs}
Let $\cP = \{p_1,\dots,p_n\}$ be a set of polynomials, and let $S = \{x \in \R^n | \forall p \in \cP: p(x) = 0\}$. We define a proof of membership in $\langle \cP\rangle$ as follows:
\begin{definition}
We say that $r(x)$ has a \emph{derivation} from $\cP$ if there is a polynomial identity of the form
\[r(x) = \sum_{i}^n \lambda_i(x) p_i(x).\]
We say that the proof has degree $d$ if $\max_i \{\deg \lambda_i p_i\} = d$.
\end{definition}

The above proof system is useful for proving when polynomials are zero on $S$, but often we want to prove that they are positive. To that end, let $\cP = \{p_1,\dots,p_n\}$ and $\cQ = \{q_1,\dots,q_m\}$ be two sets of polynomials, and let $S = \{x \in \R^n | \forall p \in \cP: p(x) = 0, \forall q \in \cQ: q(x) \geq 0\}$. We define a proof of non-negativity as follows:
\begin{definition}
 We say that $r(x)$ has a \emph{Sum-of-Squares proof of non-negativity from $\cP$ and $\cQ$} if there is a polynomial identity of the form
\[r(x) = \sum_{i}^{t_0} h_i^2(x) + \sum_{i}^m \left(\sum_{j}^{t_i} s_j^2(x)\right)q_i(x) + \sum_{i}^n \lambda_i(x) p_i(x).\]
We say the proof has degree $d$ if $\max \{\deg h_i^2, \deg s_j^2q_i, \deg \lambda_i p\} = d$.
\end{definition}
The idea behind this terminology is that if such a proof exists, then $r$ must be non-negative on $S$ since the first two terms are non-negative, and the last term is zero. We will be concerned with not just the degree of these proofs, but also their bit complexity. To this end, we define the following norms on polynomials and proofs: For $p(x) \in \R[x]$, we write $\|p\|$ for the absolute value of the maximum coefficient of $p$ in the standard monomial basis, and for any collection of polynomials $\cP$, we write $\|\cP\| = \max_{p \in \cP} \|p\|$. We will also require a bound on $\|S\| = \max_{\alpha \in S} \|\alpha\|$. Throughout this paper we will assume that the solutions $\alpha$ are {\it explicitly bounded} by $\|\alpha\| \leq 2^{\poly(n^d)}$.

\subsection{Rich Solution Spaces}
In this section we define the conditions we require in order to guarantee that SoS proofs from $\cP$ and $\cQ$ have low bit-complexity.
For a set of assignments $S$ to a polynomial system $(\cP,\cQ)$, define the moment matrix as
\[M_S \coloneqq E_{\alpha \in S}[{\bf v}(\alpha){\bf v}(\alpha)^T]\mcom\]
where the expectation is over the uniform distribution over $S$.  We will omit the subscript and write $M$, if $S$ is clear from the context.

\begin{definition}\label{def:nice}
With the above definitions, 
\begin{itemize}
\item We say that $S$ is \emph{$\delta$-spectrally rich for $(\cP, \cQ)$ up to degree $d$} if every nonzero eigenvalue of $M_S$ is at least $\delta$.
\item We say that $(\cP, \cQ)$ is \emph{$k$-complete on $S$ up to degree $d$} if every zero eigenvector $c$ of $M_S$ (which can be seen as a degree $d$ polynomial $c^T{\bf v}$) has a degree $k$ derivation from $\cP$. 
\item We say that $S$ is \emph{$\epsilon$-robust for $\cQ$} if $\forall q \in \cQ, \forall \alpha \in S: q(\alpha) > \epsilon$.
\end{itemize}
\end{definition}
Spectral richness of the solutions $S$ is equivalent to requiring if $p(x)$ is small on $S$, then there is a polynomial $q$ which agrees with $p$ on $S$ and that has small coefficients. If $(\cP, \cQ, S)$ satisfies all three conditions then we say that $S$ is $(\delta, k, \epsilon)$-\emph{rich} for $(\cP, \cQ)$ up to degree $d$. If $1/\delta = 2^{\poly(n^d)}$, $k = O(d)$, and $1/\epsilon = 2^{\poly(n^d)}$ we simply say $S$ is rich for $(\cP, \cQ)$. We choose these bounds because \prettyref{thm:main} will imply that any constraints with a rich solution space has proofs of non-negativity that can be taken to have polynomial bit complexity. Before we get into the proof of the main theorem, we exhibit polynomial systems that admit rich solutions.

\section{Examples with Rich Solution Spaces}
\label{sec:nicespaces}
In this section we present examples of polynomial systems that admit rich solution spaces. First, we consider the case $S \subseteq \{0,1\}^n$. In this case, the spectral richness is a consequence of the following easy observation.
\begin{lemma} \label{lem:integer}
	Let $M \in \R^{N \times N}$ be an integer matrix with $|M_{ij}| \leq B$ for all $i,j \in [N]$.  The smallest non-zero eigenvalue of $M$ is at least 
	$(BN)^{-N}$.
\end{lemma}
\begin{proof}
Let $A$ be a full-rank principal minor of $M$ and w.l.o.g. let it be at the upper-left block of $M$. We claim the least eigenvalue of $A$ lower bounds the least nonzero eigenvalue of $M$.
Since $M$ is symmetric, there must be a $C$ such that
\[M = \left[\begin{tabular}{c} $I$ \\ $C$\end{tabular}\right]A\left[\begin{tabular}{cc} $I$ & $C^T$\end{tabular}\right].\]
Let $P = [I, C^T]$, $\rho$ be the least eigenvalue of $A$, and $x$ be a vector perpendicular to the zero eigenspace of $A$. Then we have $x^TMx \geq \rho x^TP^TPx$,
but $x$ is also perpendicular to the zero eigenspace of $P^TP$. Now $P^TP$ has the same nonzero eigenvalues as $PP^T = I + C^TC \succeq I$, and thus $x^TP^TPx \geq 1$, and so every nonzero eigenvalue of $A$ is at least $\rho$. Now $A$ is a full-rank bounded integer matrix with dimension at most $N$. The magnitude of its determinant is at least $1$ and all eigenvalues are at most $N \cdot B$.  Therefore, its least eigenvalue must be at least $(BN)^{-N}$ in magnitude. 
\end{proof}

\begin{lemma}\label{lem:integer-rich}
Let $\cP$ and $\cQ$ be such that $S \subseteq \{0,1\}^n$. Then $S$ is $\delta$-spectrally rich with $\frac{1}{\delta} = 2^{\poly(n^d)}$.
\end{lemma}
\begin{proof}
	Recall $M = E_{\alpha \in S}[{\bf v}(\alpha){\bf v}(\alpha)^T]$, and note that $|S| \cdot M$ is an integer matrix with entries at most $2^n$.  The proof follows by applying \prettyref{lem:integer}. 
\end{proof}

To prove completeness, we typically want to show two things. First, that every degree $d$ polynomial in $\langle \cP\rangle$ has a degree at most $k$ derivation. Second, that there are no polynomials outside $\langle \cP\rangle$ that are zero on $S$. This second condition can be thought of as saying that the set of equations $\cP$ is somehow maximal, i.e., if there are extra polynomial equalities implied by $\cQ$, they should be included in $\cP$.  Here we consider a few examples.

\paragraph{\textsc{Max-CSP}: $\cP = \{x_i^2 - x_i | i \in [n]\}$}
	Here $S = \{0,1\}^n$.  Any polynomial $p$ of degree $d$ can be multilinearized one monomial at a time.  Specifically, we can find degree $d$ multilinear $p^*$ such that $p - p^* = 0$ has a degree $d$ derivation from $\cP$.  Furthermore, the multilinear polynomial $p^*$ is zero over $S$ if and only if all its coefficients are zero.  Thus $\cP$ is $d$-complete up to degree $d$ for all $d \in \N$.  

\paragraph{\textsc{Max-Clique}: $\cP = \{x_i^2 - x_i | i \in [n]\} \cup \{x_ix_j | (i,j) \notin E\}$}
Here $S$ is the set of all cliques in the graph.  Suppose $p$ is a polynomial that is identically zero over $S$.  Without loss of generality, we may assume $p$ is multilinear, if otherwise we can multilinearize it using $\{ x_i^2 - x_i | i \in [n]\}$.  For a multilinear polynomial $p(x) = \sum_{\alpha \subset [n]} \hat{p}_\alpha x_{\alpha}$, we claim that if $p(x) = 0 \forall x \in S$ then for all cliques $\alpha \subset [n]$, the corresponding coefficient $\hat{p}_\alpha = 0$, i.e., all non-zero coefficients of $p$ are non-cliques.
Suppose not, then let $\alpha$ be the smallest clique with $\hat{p}_\alpha \neq 0$.  Then, we will have $p( \Ind_{\alpha}) = \hat{p}_\alpha \neq 0$, a contradiction.
Since all coefficients of $p$ are non-cliques, each monomial in $p$ can be eliminated using an appropriate polynomial from $\{ x_i x_j | (i,j) \notin E\}$.

\begin{remark}
	More generally, the above two cases are special cases of the following general setup: $\cQ$ is empty, and $\cP$ is a Gr\"obner basis. A Gr\"obner basis for an ideal is a generating set of polynomials that allow a well-defined multivariate polynomial division (see \cite{Grobner} for more information). Computing the Gr\"obner basis is often the first step in practical polynomial equation solvers, and we note the following easy lemma:
\begin{lemma}\label{lem:grobner}
If $\cQ = \emptyset$ and $\cP$ is a Gr\"obner basis for $\langle \cP\rangle$, then $S$ is $d$-complete up to degree $d$. 
\end{lemma}
\begin{proof}
If $\cP$ is a Gr\"obner basis, then every degree $d$ polynomial in $\langle \cP\rangle$ has a degree $d$ derivation via multivariate division. Because $\cQ = \emptyset$, the polynomials that are zero on $S$ are exactly the polynomials in $\langle \cP\rangle$. 
\end{proof}

\end{remark}

%
%

\paragraph{\textsc{Balanced Separator}: $\cP = \{x_i^2 - x_i | i \in [n]\}$, $\cQ = \{2n/3 - \sum_i x_i, \sum_i x_i - n/3\}$}
%
The solution space $S$ here is all bit strings with hamming weight between $n/3$ and $2n/3$. 
Suppose $r$ is a polynomial that is zero on $S$.  
Without loss of generality, we may assume that $r$ is multilinear by using the constraints $\{x_i^2 - x_i | i \in [n]\}$.
Suppose $r$ is a non-zero multilinear polynomial which is zero on $S$, then its symmetrized version $r^* = \frac{1}{n!}\sum_{\sigma \in \cS_n} \sigma r$ must also be zero on $S$, where $\sigma$ acts by permuting the variable names. However, $r^*$ is a univariate polynomial in $\sum_i x_i$ (modulo the Boolean constraints). 
This univariate polynomial has $n/3$ zeros, and thus must have degree at least $n/3$. Since symmetrizing doesn't change degree, we conclude that $r$ also has degree at least $n/3$. Thus every non-zero multilinear polynomial that is zero on $S$ but not in $\langle \cP\rangle$, has degree at least $n/3$. 
Therefore the system is $d$-complete up to degree $d$ for $d \leq \frac{n}{3}$.
%
The polynomials in $\cQ$ can be perturbed by $1/2$ to make them $1/2$-robust, and thus $S$ is rich for $(\cP, \cQ)$. 

\paragraph{\textsc{Matching}: $\cP = \{x_{ij}^2 - x_{ij} | i,j \in [n]\} \cup \{\sum_i x_{ij} - 1 | i \in [n]\} \cup \{x_{ij}x_{ik} | i,j,k \in [n]\}$} These constraints are $2d$-complete as proven in \cite{Braun:2016:MPN:2884435.2884510}.

\paragraph{\textsc{Max-Bisection}: $\cP = \{x_i^2 - x_i | i \in [n]\} \cup \{\sum_i x_i - n/2\}$} We will prove in \prettyref{sec:balance} that these constraints are $d$-complete. The proof will be very similar to the one for \textsc{Matching}, due to the similar symmetry of the constraints.

\paragraph{\textsc{Unit-Vector}: $\cP = \{\sum_i x_i^2 - 1\}$} Here $S = \{x: \|x\| = 1\}$. This constraint appears frequently in tensor norm problems as a way to enforce scaling. Since $\cQ = \emptyset$, it is clearly robust. It may be well-known that $\cP$ is $d$-complete, but we could not find a reference so we record it here for completeness. Let $p(x)$ be any degree $d$ polynomial which is zero on the unit sphere, and define $p_0(x) = p(x) + p(-x)$. Clearly $p_0$ is also zero on the unit sphere, with degree $k = 2\lfloor (d+1)/2 \rfloor$. Note that $p_0$ has only terms of even degree. 
Define a sequence of polynomials $\{p_i\}_{i \in \{0,\ldots, k\}}$ as follows.
Define $q_i$ to be the part of $p_i$ which has degree strictly less than $k$, and let $p_{i+1} = p_i + q_i\cdot(\sum_i x_i^2 - 1)$. Then each $p_i$ is zero on the unit sphere and has no monomials of degree strictly less than $2i$. Thus $p_{k/2}$ is homogeneous of degree $k$. But then $p(tx) = t^kp_k(x) = 0$ for any unit vector $x$ and $t > 0$, and thus $p_k(x)$ must be the zero polynomial. This implies that $p_0$ is a multiple of $\sum_i x_i^2 - 1$. The same logic shows that $p(x) - p(-x)$ is also a multiple of $\sum_i x_i^2 - 1$, and thus so is $p(x)$. Now $\langle \cP\rangle$ is principal, so every degree $d$ polynomial in it has a degree $d$ derivation, so $(\cP, \cQ, S)$ is $d$-complete.

To prove spectral-richness, we note that in \cite{10.2307/2695802} the author gives an exact formula for each entry of the matrix $M = \int_{S} p(x)$ for any polynomial $p$. The formulas imply that $(n+d)!\pi^{-n/2} M$ is an integer matrix with entries (very loosely) bounded by $(n+d)!d!2^n$. By \prettyref{lem:integer}, we conclude that $S$ is $\delta$-spectrally rich with $1/\delta = 2^{\poly(n^d)}$.

We collect the examples discussed in this section here:
\begin{corollary}\label{cor:examples}
The following constraints admit rich solutions:
\begin{itemize}
\item \textsc{Max-CSP}: $\cP = \{x_i^2 - x_i | i \in [n]\}$. 
\item \textsc{Max-Clique}: $\cP = \{x_i^2 - x_i | i \in [n]\} \cup \{x_ix_j | (i,j) \notin E\}$.
\item \textsc{Balanced Separator}: $\cP = \{x_i^2 - x_i | i \in [n]\}$, $\cQ = \{2n/3 - \sum_i x_i, \sum_i x_i - n/3\}$.
\item \textsc{Matching}: $\cP = \{x_{ij}^2 - x_{ij} | i,j \in [n]\} \cup \{\sum_i x_{ij} - 1 | i \in [n]\} \cup \{x_{ij}x_{ik} | i,j,k \in [n]\}$.
\item \textsc{Max-Bisection}: $\cP = \{x_i^2 - x_i | i \in [n]\} \cup \{\sum_i x_i - n/2\}$.
\item \textsc{Unit-Vector}: $\cP = \{\sum_i x_i^2 - 1\}$.
\end{itemize}
\end{corollary}

\subsection{Limitations}
While \prettyref{thm:main} allows us to prove that many different systems of polynomial constraints have well-behaved SoS proofs, there are a few areas where it comes up short. Most noticeably, to contain a rich set of solutions the solution space has to be nonempty. This can be a problem when trying to find SoS proofs of infeasibility. For example, one common technique is to introduce lower bounds on an objective function $f(x)$ of a maximization problem as constraints and attempt to use SoS to find a refutation, i.e. a proof of non-negativity for the constant polynomial $-1$. We are unable to show that these proofs can be taken to have polynomial bit complexity since they have empty solution spaces. As another example, we are unable to use our framework to show that refutations of the knapsack constraints use only polynomially many bits, even though it is clear by simply examining these known refutations that they only involve small coefficients.

\section{Rich Solution Spaces Yield Bounded SoS Proofs}
\label{sec:main}

In this section we prove our main theorem:

\begin{theorem}\label{thm:main}
Let $\cP = \{p_1,\dots,p_m\}$ and $\cQ = \{q_1, \dots, q_\ell\}$ be sets of polynomials with $S = \{\alpha \in \R^n | \forall p \in \cP: p(\alpha) = 0, \forall q \in \cQ: q(\alpha) \geq 0\}$. Assume that the set of solutions $S$ is $(k,\delta,\epsilon)$-rich for $(\cP,\cQ)$.

Let $r(x)$ be a polynomial nonnegative on $S$, and assume $r$ has a degree $d$ sum-of-squares proof of nonnegativity 
\[r(x) = \sum_{i=1}^{t_0} h_i^2 + \sum_{i=1}^\ell \left(\sum_{j=1}^{t_i} s_j^2\right) q_i + \sum_{i=1}^m \lambda_i p_i.\] 
Then $r$ has a degree $k$ sum-of-squares proof of nonnegativity such that the coefficients of every polynomial appearing in the proof are bounded by $2^{\poly(n^k,\log \frac{1}{\delta},\log \frac{1}{\epsilon})}$. In particular, if $S$ is rich then every coefficient can be written down with only $\poly(n^d)$ bits. 
\end{theorem}
\begin{proof}
First, we rewrite the proof into a more convenient form before proving bounds on each individual term. Because the elements of ${\bf v}$ are a basis for $\R[x]_d$, every polynomial in the proof can be expressed as $c^T{\bf v}$, where $c$ is a vector of reals:
\begin{align*} r(x) &= \sum_{i=1}^{t_0} (c_i^T{\bf v})^2 + \sum_{i=1}^\ell \left(\sum_{j=1}^{t_i} (d_{ij}^T{\bf v})^2\right)q_i + \sum_{i=1}^m \lambda_i p_i \\
&= \langle C, {\bf v}{\bf v}^T\rangle + \sum_{i=1}^\ell \langle D_i, {\bf v}{\bf v}^T\rangle q_i + \sum_{i=1}^m \lambda_i p_i
\end{align*}
for PSD matrices $C$, $D_1,\dots,D_\ell$. Next, we average this polynomial identity over all the points $\alpha \in S$:
\begin{align*}
E_{\alpha \in S}[r(\alpha)] &= \langle C, E_{\alpha \in S}[{\bf v}(\alpha){\bf v}(\alpha)^T]\rangle + \sum_{i=1}^\ell \langle D_i, E_{\alpha \in S}[{\bf v}(\alpha){\bf v}(\alpha)^T]\rangle q_i(\alpha) + 0 \\
\end{align*}
The LHS is at most $\poly(\|r\|, \|S\|)$, and the RHS is a sum of positive numbers, so the LHS is a bound on each term of the RHS. 
We would like to say that since $S$ is $\delta$-spectrally rich, the first term is at least $\delta Tr(C)$. 
Unfortunately the averaged matrix may have zero eigenvectors, and it is possible that $C$ could have very large eigenvalues in these directions. 
However these eigenvectors must correspond to polynomials that are zero on $S$. Because $(\cP, \cQ, S)$ is complete, these can be absorbed into the final term. More formally, let $\Pi = \sum_u uu^T$ be the projector onto the zero eigenspace of $M = E_{\alpha \in S}[{\bf v}(\alpha){\bf v}(\alpha)^T]$. Because $(\cP, \cQ, S)$ is complete, for each $u$ we have a degree $k$ derivation $u^T{\bf v} = \sum_i \sigma_{ui} p_i$. Then $\Pi {\bf v}{\bf v}^T = \sum_u (u^T{\bf v}) u{\bf v}^T$. Thus we can write
\begin{align*}
\langle C, {\bf v}{\bf v}^T\rangle &= \langle C, (\Pi + \Pi^\perp){\bf v}{\bf v}^T(\Pi + \Pi^\perp)\rangle \\
&= \langle C, \Pi^\perp {\bf v}{\bf v}^T \Pi^\perp\rangle + \sum_u u^T{\bf v}\left(\langle C, \Pi^\perp {\bf v}u^T + {\bf v}u^T\Pi^\perp + {\bf v}u^T\Pi\rangle\right) \\
&= \langle \Pi^\perp C \Pi^\perp, {\bf v}{\bf v}^T\rangle + \sum_i \sigma_i p_i.
\end{align*}
Doing the same for the other terms and setting $C' = \Pi^\perp C \Pi^\perp$ and similarly for $D_i'$, we get a new proof:
\[r(x) = \langle C', {\bf v}{\bf v}^T\rangle + \sum_{i=1}^\ell \langle D_i', {\bf v}{\bf v}^T\rangle q_i + \sum_{i=1}^m \lambda_i' p_i.\]
Now after averaging over $S$, the zero eigenspace of $C'$ is contained in the zero eigenspace of $M$. Taken with the $\delta$-spectral richness, we have
\[\poly(\|r\|,\|S\|) \geq \delta Tr(C) + \sum_{i=1}^\ell \delta Tr(D_i') q_i(\alpha).\]
Because each $q_i(\alpha) \geq \epsilon$, we get $C'$ and $D_i'$ have entries bounded by $\poly(\|r\|, \|S\|, \frac{1}{\delta}, \frac{1}{\epsilon})$.

The only thing left to do is to bound the coefficients $\lambda_i'$, but this is easy because the SoS proof is linear in these coefficients. If we imagine the coefficients of the $\lambda_i'$ as variables, then the linear system induced by the polynomial identity
\[r(x) - \langle C', {\bf v}{\bf v}^T\rangle - \sum_{i=1}^\ell \langle D_i', {\bf v}{\bf v}^T\rangle = \sum_{i=1}^m \lambda_i' p_i\]
is clearly feasible, and the coefficients of the LHS are bounded by $\poly(\|r\|, \|S\|, \frac{1}{\delta}, \frac{1}{\epsilon})$. There are $O(n^k)$ variables, so by Cramer's rule, the coefficients of the $\lambda_i'$ can be taken to be bounded by $\poly(\|\cP\|^{n^k}, \frac{1}{\delta}, \frac{1}{\epsilon}, \|r\|, \|S\|, n!)$. $\|\cP\|, \|r\| \leq 2^{\poly(n^d)}$ as they are considered part of the input, $\|S\| \leq 2^{\poly(n^d)}$ by the explicitly bounded assumption, and $d \leq k$. Thus, this bound is at most $2^{\poly(n^k, \log \frac{1}{\delta}, \log \frac{1}{\epsilon})}$.

\end{proof}

\section{Boolean Systems With No Small-Coefficient Proofs}\label{sec:counterexample}
In \cite{odonnell17}, the author gives an example of a polynomial system for which degree two SoS proofs can certify non-negativity of a certain poylnomial, but the proofs necessarily involves coefficients of doubly-exponential size. However, there are two weaknesses in his example system. First, it is not a Boolean one, i.e. it contains variables $y_i$ for which the constraint $y_i^2 - y_i = 0$ is not present in the constraints. Many practical optimization problems have Boolean constraints, and in \cite{odonnell17}, the author hoped that having those constraints might suffice to imply that all proofs could have small bit complexity. Second, while the degree two proofs must have exponential bit complexity, there were degree four proofs of non-negativity with polynomial bit complexity. 
%
%
In this section, we strengthen his counterexample, giving an example of a Boolean system with $n$ variables for which there is a polynomial that has a degree two proof of non-negativity, but no proof with polynomial bit complexity until degree $\Omega(\sqrt{n})$.

\subsection{A First Example}
The original example given in \cite{odonnell17} essentially contains the following system whose repeated squaring is responsible for the blowup of the coefficients in the proofs:
\[\begin{tabular}{ccccc}
$y_1^2 - y_2 = 0$, & $y_2^2 - y_3 = 0$, & $\dots$, & $y_{n-1}^2 - y_n = 0$, & $y_n^2 = 0$.
\end{tabular}\]
Clearly, the only solution to the system is $(0,0,0,\dots,0)$, and therefore the polynomial $\epsilon - y_1$ must be non-negative over the solution space for any $\epsilon > 0$. It is not as obvious whether or not an SoS proof of this non-negativity exists. It turns out that there is a degree two SoS proof as follows:
\begin{align}
\epsilon - y_1 &\equiv \left(\sqrt{\frac{\epsilon}{n}} - \left(\frac{n}{4\epsilon}\right)^{1/2}y_1\right)^2 + \left(\sqrt{\frac{\epsilon}{n}} - \left(\frac{n}{4\epsilon}\right)^{3/2}y_2\right)^2 + \left(\sqrt{\frac{\epsilon}{n}} - \left(\frac{n}{4\epsilon}\right)^{7/2}y_3\right)^2 + \nonumber\\
&+\dots + \left(\sqrt{\frac{\epsilon}{n}} - \left(\frac{n}{4\epsilon}\right)^{(2^n-1)/2}y_n\right)^2.\label{eq:proof}\tag{$*$}
\end{align}
where the $\equiv$ is equality modulo the ideal generated by the constraints. Of course, this proof involves coefficients of doubly-exponential size, but one can prove that they are required. We will take $\epsilon < 1/2$ for simplicity. We will define a linear functional $\phi: \R[Y]_d \rightarrow \R$ satisfying the following:
\begin{itemize}
\item $\phi[\epsilon - y_1] = -\epsilon$
\item $\phi[p^2] \geq 0$ for any $p^2$ of degree at most $d$
\item $\phi[\sigma_i(y_i^2 - y_{i+1})] = 0$ for any $i \leq n-1$ and $\sigma_i$ of degree at most $d-2$
\item $|\phi[\lambda y_n^2]| \leq (2\epsilon)^{2^{n-1}}n^d\|\lambda\|$.
\end{itemize}
If such a $\phi$ exists, then for any degree $d$ SoS proof of non-negativity
\[\epsilon - y_1 = \sum_i h_i(y)^2 + \sum_{i=1}^{n-1} \sigma_i(y_i^2 - y_{i+1}) + \lambda \cdot y_n^2,\]
apply $\phi$ to both sides. We obtain $-\epsilon \leq P + 0 + \phi[\lambda y_n^2]$, where $P \geq 0$. Because $|\phi[\lambda y_n^2]| \leq (2\epsilon)^{2^{n-1}}n^d\|\lambda\|$, $\lambda$ must contain a coefficient of size at least $\Omega(\frac{1}{n^d}\left(\frac{1}{2\epsilon}\right)^{2^n})$.

To show that such a $\phi$ exists, we define it as follows. By the constraints, every monomial is equivalent to some power of $y_1$. For example, $y_1y_2y_3 \equiv y_1^7$.  More generally, the constraints imply that $\prod_{i = 1}^n y_i^{\beta_i} = y_1^{\sum_{j = 1}^n 2^{j-1} \beta_j}$.  Define $\phi$ by,

\[\phi\left( \prod_{i = 1}^n y_i^{\beta_i}\right) = (2\epsilon)^{\sum_{i} 2^{i-1} \beta_i } \]

One can easily check that this $\phi$ satisfies the above. Note that none of the variables $y_i$ in the above system are boolean, which we achieve in the upcoming section.
 
\subsection{A Boolean System}
One simple way to try to make the system Boolean is to just add the constraints $y_i^2 = y_i$ to the system. Unfortunately, in that case it is easy to prove that $y_i - y_j = 0$ for each $i$ and $j$, and of course $y_n = y_n^2 = 0$. It is too easy for SoS to figure out what each $y_i$ should look like. Previously, the variables were unconstrained in any way, and we want to imitate that. We draw inspiration from the Knapsack problem, and we instead replace each instance of the variable $y_i$ with a sum of $2k$ Boolean variables 
\[y_i \rightarrow \sum_j w_{ij} - k,\] 
and we consider the non-negative polynomial $\epsilon - (\sum_j w_{1j} - k)$. Clearly there is a degree two proof of non-negativity for this polynomial since we can just replace each instance of $y_i$ with $\sum_j w_{ij} - k$ in (\ref{eq:proof}). 

It remains to show that there are no other proofs that have only small coefficients. Here, we use the fact that the Knapsack problem is hard for SoS: there is no SoS proof of degree less than $\Omega(k)$ that $\sum_j w_{ij} - k$ is not equal to any number $r \in (0,1)$ \cite{Grigoriev2001}. This allows us to use the Knapsack pseudodistribution to "pretend" that $\sum_j w_{ij} - k = (2\epsilon)^{2^{i-1}}$. Specifically, for each $r \in (0,1)$, there is a linear functional $\phi_r$ defined on polynomials of $2k$ Boolean variables which satisfies
\begin{itemize}
\item $\phi_r[\sigma_{ij}(w_{ij}^2 - w_{ij})] = 0$ for any $\sigma_{ij}$ up to degree $O(k)$
\item $\phi_r[\lambda\cdot((\sum_j w_{ij} - k) - r)] = 0$ for any polynomial $\lambda$ up to degree $O(k)$
\item $\phi_r[p^2] \geq 0$ for any polynomial $p^2$ of degree at most $O(k)$.
\end{itemize}
Now, take the linear functional $\Phi$ defined on each polynomials of $2kn$ variables defined in the following way: Let $T = T_1 \cup T_2 \cup \dots \cup T_n$ where $T_i$ is a multiset that contains only the variables corresponding to $y_i$, and let $w_T$ denote the associated monomial. Then define
\[\Phi[w_T] = \phi_{2\epsilon}(w_{T_1})\phi_{(2\epsilon)^2}(w_{T_2})\dots\phi_{(2\epsilon)^{2^{n-1}}}(w_{T_n}).\]
Clearly $\Phi$ is non-negative on squares and $\Phi[\sigma_{ij}(w_{ij}^2-w_{ij})] = 0$ for any $\sigma_{ij}$ up to degree $\Omega(k)$. Because $\Phi[\lambda(\sum_j w_{ij} - k)] = \Phi[(2\epsilon)^{2^{i-1}}\lambda]$, $\Phi$ also satisfies $\Phi[\lambda((\sum_j w_{ij}-k)^2 - (\sum_j w_{i+1,j} - k))] = 0$ for each $\lambda$ and $1 \leq i \leq n-1$. Finally, because each variable is Boolean, $\Phi$ of any monomial is at most one, so for any monomial $w_M$, $\Phi[w_M(\sum_j w_{nj} - k)^2] = \Phi[(2\epsilon)^{2^{n-1}} w_M] \leq (2\epsilon)^{2^{n-1}}$. There are at most $(nk)^d$ monomials, so $\Phi[\lambda(\sum_j w_{nj} - k)^2] \leq (nk)^d(2\epsilon)^{2^{n-1}}\|\lambda\|$. Just as before, the existence of $\Phi$ implies that any degree $d$ proof of non-negativity for $\epsilon - (\sum_j w_{1j} - k)$ must contain coefficients of size at least $\Omega(\frac{1}{(nk)^d} \cdot \left(\frac{1}{2\epsilon}\right)^{2^n})$. If we set $k = n$, then there are $n^2$ variables and no proof of non-negativity with coefficients smaller than doubly-exponential until degree $n$. This proves \prettyref{thm:counter}.

\section{\textsc{Max-Bisection} Constraints}\label{sec:balance}
In this section, we prove our earlier claim that the \textsc{Max-Bisection} constraints admit rich solutions. Recall the constraints:
\[\cP(n) = \left\{x_i^2 - x_i | i \in [2n]\right\} \cup \left\{\sum_i x_i - n\right\}.\]
Recall that to prove $S$ is rich, we have to prove that it is spectrally rich, robust, and complete. Since the solution space lies in the hypercube, it is spectrally rich by \prettyref{lem:integer-rich}, and it is clearly robust since $\cQ$ is empty. It remains to prove that it is complete for some $k$. This proof follows a very similar path to \cite{Braun:2016:MPN:2884435.2884510}, due to the similar symmetry of the constraints. 
\begin{lemma}
$\cP(n)$ is $d$-complete for any $d \leq n$.
\end{lemma}
\begin{proof}
Let $S(n)$ denote the solution space of $\cP(n)$, and let $M = \E_{\alpha \in S}[{\bf v}(\alpha){\bf v}(\alpha)^T]$. Any zero eigenvector $c$ of $M$ can be associated with a polynomial $c^T{\bf v}$. Since 
\[c^TMc = \sum_\alpha (c^T{\bf v}(\alpha))^2 = 0,\]
we must have $c^Tv(\alpha) = 0$ for each $\alpha \in S$. We argue that any degree $d$ polynomial which is identically zero on $S(n)$ must have a degree $d$ derivation from $\cP(n)$. 

We proceed by induction on $d$. If $d = 0$, the only constant polynomial zero on $S(n)$ is the zero polynomial, which has the trivial derivation. Now consider the case of $d = c+1$. We proceed in two parts. First, if $r$ is fully symmetric, we show that it has a degree $d$ derivation. Secondly, for any polynomial $p$ which is zero on $S(n)$, we prove that $p - \frac{1}{(2n)!}\sum_{\sigma \in \cS_n} \sigma p$ has a degree $d$ derivation from $\cP$, where $\sigma$ acts on $p$ by permuting the labels of the variables. Taken together, these two facts imply that $r$ has a degree $d$ derivation from $\cP(n)$.

To prove the first part, note that a symmetric polynomial $r$ is a linear combination of the elementary symmetric polynomials $e_1,\dots,e_c$, and it is clear that $e_k(x)$ can be derived by taking the polynomial $(\sum_i x_i - n)^k$, reducing it to multilinear using the boolean constraints, and then reducing by $e_l(x)$ for each $l < k$. This will result in a constant polynomial, which must be the zero polynomial since we are only adding polynomials which are zero on $S(n)$, so the resulting polynomial must be zero on $S(n)$. 

To prove the second part, let $\sigma_{ij}$ be the transposition of labels $i$ and $j$, and consider the polynomial $r - \sigma_{ij}r$. Writing $r = r_ix_i + r_jx_j + r_{ij}x_ix_j + q_{ij}$, where none of $r_i$,$r_j$,$r_{ij}$, nor $q_{ij}$ depend on $x_i$ or $x_j$, we can rewrite
\[r - \sigma_{ij}r = (r_i - r_j)(x_i - x_j).\]
Now because $r - \sigma_{ij}r$ evalutes to zero on any boolean string with exactly $n$ ones, if we set $x_i = 1$ and $x_j = 0$, we know that $r_i - r_j$ is a polynomial that must evaluate to zero on any boolean string with exactly $n-1$ ones. Because $\deg (r_i - r_j) = d-1$, by the inductive hypothesis, $r_i - r_j$ has a degree $d-1$ proof from $\cP(n-1)$ (since $d \leq n$, clearly $d-1 \leq n-1$). This implies that $(r_i - r_j)(x_i - x_j)$ has a degree $d-1$ proof from $\cP(n)$:
\begin{align*}
(r_i - r_j)(x_i - x_j) &= \left[\sum_{t \neq i,j} \lambda_t\cdot (x_t^2 - x_t) + \lambda \cdot \left(\sum_{t \neq i,j} x_t - (n-1)\right)\right](x_i - x_j) \\
&= \sum_{t} \lambda'_t \cdot (x_t^2 - x_t) + \lambda \cdot \left(\sum_{t \neq i,j} x_t - (n-1) + (x_i + x_j - 1)\right)(x_i - x_j)\\
&= \sum_t \lambda'_t \cdot (x_t^2 - x_t) +\lambda' \cdot \left(\sum_t x_t - n\right)
\end{align*}
where we used the fact that $(x_i + x_j - 1)(x_i - x_j) - (x_i^2 - x_i) + (x_j^2 - x_j) = 0$. The degree of this derivation is at most $d$ because each $\lambda_t$ has degree at most $d-3$, and $\lambda'_t = \lambda_t(x_i - x_j)$, and similarly for $\lambda$. Thus the inductive hypothesis implies that $r - \sigma_{ij}r$ has a degree $d$ derivation, and since transpositions generate the symmetric group, this implies that $r - \frac{1}{(2n)!}\sum_{\sigma \in \cS_n} \sigma r$ has a degree $d$ proof from $\cP(n)$.
\end{proof}

\begin{remark}
In this example, $\cP$ is not a Gr\"obner basis for its ideal $\langle \cP \rangle$. 
Indeed, the Gr\"obner basis for this ideal has exponential size. This is an example where our framework is applicable, even though Gr\"obner bases are intractable to compute. 
\end{remark}


\addreferencesection
\bibliographystyle{amsalpha}
\bibliography{files/writeup,files/pc}

\newcommand{\etalchar}[1]{$^{#1}$}
\providecommand{\bysame}{\leavevmode\hbox to3em{\hrulefill}\thinspace}
\providecommand{\MR}{\relax\ifhmode\unskip\space\fi MR }
\providecommand{\MRhref}[2]{%
  \href{http://www.ams.org/mathscinet-getitem?mr=#1}{#2}
}
\providecommand{\href}[2]{#2}
\begin{thebibliography}{BBCH{\etalchar{+}}16}

\bibitem[AL94]{Grobner}
William Adams and Philippe Loustaunau, \emph{An introduction to gr\"obner
  bases}, American Mathematical Society, 1994.

\bibitem[BBCH{\etalchar{+}}16]{Braun:2016:MPN:2884435.2884510}
G\'{a}bor Braun, Jonah Brown-Cohen, Arefin Huq, Sebastian Pokutta, Prasad
  Raghavendra, Aurko Roy, Benjamin Weitz, and Daniel Zink, \emph{The matching
  problem has no small symmetric sdp}, Proceedings of the Twenty-seventh Annual
  ACM-SIAM Symposium on Discrete Algorithms (Philadelphia, PA, USA), SODA '16,
  Society for Industrial and Applied Mathematics, 2016, pp.~1067--1078.

\bibitem[BBH{\etalchar{+}}12]{BBHKSZ12}
Boaz Barak, Fernando~G.S.L. Brandao, Aram~W. Harrow, Jonathan Kelner, David
  Steurer, and Yuan Zhou, \emph{Hypercontractivity, sum-of-squares proofs, and
  their applications}, Proceedings of the Forty-fourth Annual ACM Symposium on
  Theory of Computing (New York, NY, USA), STOC '12, ACM, 2012, pp.~307--326.

\bibitem[BRS11]{barak2011rounding}
B.~Barak, P.~Raghavendra, and D.~Steurer, \emph{Rounding semidefinite
  programming hierarchies via global correlation}, Proc.\ FOCS, IEEE, 2011,
  pp.~472--481.

\bibitem[BS14]{barak2014sum}
Boaz Barak and David Steurer, \emph{Sum-of-squares proofs and the quest toward
  optimal algorithms}, In Proceedings of the 2014 International Congress of
  Mathematicians. International Mathematical Union (2014).

\bibitem[Fol01]{10.2307/2695802}
Gerald~B. Folland, \emph{How to integrate a polynomial over a sphere}, The
  American Mathematical Monthly \textbf{108} (2001), no.~5, 446--448.

\bibitem[Gri01]{Grigoriev2001}
D.~Grigoriev, \emph{Complexity of positivstellensatz proofs for the knapsack},
  computational complexity \textbf{10} (2001), no.~2, 139--154.

\bibitem[GS11]{GuruswamiS11}
Venkatesan Guruswami and Ali~Kemal Sinop, \emph{Lasserre hierarchy, higher
  eigenvalues, and approximation schemes for graph partitioning and quadratic
  integer programming with {PSD} objectives}, {IEEE} 52nd Annual Symposium on
  Foundations of Computer Science, {FOCS} 2011, Palm Springs, CA, USA, October
  22-25, 2011, 2011, pp.~482--491.

\bibitem[GV01]{grigoriev2001complexity}
Dima Grigoriev and Nicolai Vorobjov, \emph{Complexity of null-and
  positivstellensatz proofs}, Annals of Pure and Applied Logic \textbf{113}
  (2001), no.~1-3, 153--160.

\bibitem[Las00]{lasserre2000optimisation}
Jean~Bernard Lasserre, \emph{Optimisation globale et th{\'e}orie des moments},
  Comptes Rendus de l'Acad{\'e}mie des Sciences-Series I-Mathematics
  \textbf{331} (2000), no.~11, 929--934.

\bibitem[Las01]{lasserre2001global}
Jean~B Lasserre, \emph{Global optimization with polynomials and the problem of
  moments}, SIAM Journal on Optimization \textbf{11} (2001), no.~3, 796--817.

\bibitem[Lau09]{laurent2009sums}
Monique Laurent, \emph{Sums of squares, moment matrices and optimization over
  polynomials}, Emerging applications of algebraic geometry, Springer, 2009,
  pp.~157--270.

\bibitem[Nes00]{nesterov2000squared}
Yurii Nesterov, \emph{Squared functional systems and optimization problems},
  High performance optimization, Springer, 2000, pp.~405--440.

\bibitem[Oâ€17]{odonnell17}
Ryan O'Donnell, \emph{Sos is not obviously automatizable, even
  approximately}, Innovations in Theoretical Computer Science (ITCS) (2017).

\bibitem[Par00]{parrilo2000structured}
Pablo~A Parrilo, \emph{Structured semidefinite programs and semialgebraic
  geometry methods in robustness and optimization}, Ph.D. thesis, California
  Institute of Technology, 2000.

\bibitem[Sho87]{shor1987class}
Naum~Z Shor, \emph{Class of global minimum bounds of polynomial functions},
  Cybernetics \textbf{23} (1987), no.~6, 731--734.

\end{thebibliography}

\appendix

\end{document}
